\documentclass[12pt,pra,aps,amssymb,amsfonts,amsmath,tightenlines]{revtex4}
\usepackage{dsfont}
\usepackage{graphicx}
\usepackage{amssymb,amsfonts,amsthm}
\usepackage{color}
\usepackage{verbatim}
\usepackage[normalem]{ulem}
\usepackage{subfigure}
\newcommand{\half}{\mbox{$\textstyle \frac{1}{2}$}}

\newtheorem{thm}{Theorem}

\newtheorem{rem}[thm]{Remark}
\newtheorem{prop}[thm]{Proposition}

\begin{document}
\newcommand{\p}[1]{\ensuremath{\mathbb{#1}}}
\newcommand{\gb}[1]{\ensuremath{\gamma_{#1 T}}}
\newcommand{\gbb}[1]{\ensuremath{\bar{\gamma}_{#1 T}}}
\newcommand{\levy}{{L\'evy }}
\newcommand{\D}[1]{\ensuremath{\Delta_{#1}}}
\newcommand{\tD}[1]{\ensuremath{\Delta^{\pi}_{#1}}}
\def\1{\mathds{1}}
\def\R{\mathbb{R}}
\def\P{\mathbb{P}}
\def\Q{\mathbb{Q}}
\def\E{\mathbb{E}}
\def\N{\mathbb{N}}
\def\d{\,\mathrm{d}}
\def\dd{\mathrm{d}}
\def\half{\frac{1}{2}}
\def\var{\mathrm{Var}}
\def\cov{\mathrm{Cov}}
\def\th{\theta}
\def\e{\mathrm{e}}
\def\b{\beta}
\def\g{\gamma}
\def\G{\Gamma}
\def\a{\alpha}
\def\ba{\boldsymbol\alpha}
\def\s{\sigma}
\def\l{\lambda}
\def\law{\overset{\textrm{law}}{=}}
\def\tp{\mathrm{T}}
\def\F{\mathcal{F}}
\def\lrb{\mathit{LRB}}
\def\lrbc{\mathit{LRB}_{\mathcal{C}}}
\def\lrbd{\mathit{LRB}_{\mathcal{D}}}
\def\i{\mathrm{i}}
\def\Borel{\mathcal{B}(\R)}
\def\SHS{stable-$\half$ subordinator}
\def\SH{stable-$\half$ }
\def\k{\kappa}
\def\t{\tau}


\title{Stable-$\bf \half$ Bridges and Insurance}

\author{Edward Hoyle$^*$, Lane~P.~Hughston$^{\dagger **}$ and  Andrea Macrina$^{**\mathsection}$}

\affiliation{$^*$Fulcrum Asset Management, 5--7 Chesterfield Gardens, London W1J 5BQ, UK\\ 
 $^{\dagger}$Department of Mathematics, Brunel University London, Uxbridge UB8 3PH, UK \\
 $^{**}$Department of Mathematics, University College London, 
London WC1E 6BT, UK\\ $^{\mathsection}$Department of Actuarial Science, University of Cape Town, Rondebosch 7701, RSA} 
 \date{9 April 2014}

\begin{abstract}
We develop a class of non-life reserving models using a stable-$\half$ random bridge to simulate the accumulation of paid claims, allowing for an essentially arbitrary choice of \emph{a priori} distribution for the ultimate loss.
Taking an information-based approach to the reserving problem, we derive the process of the conditional distribution of the ultimate loss.
The ``best-estimate ultimate loss process" is given by the conditional expectation of the ultimate loss.
We derive explicit expressions for the best-estimate ultimate loss process, and for expected recoveries arising from aggregate excess-of-loss reinsurance treaties.
Use of a deterministic time change allows for the matching of any initial (increasing) development pattern for the paid claims.
We show that these methods are well-suited to the modelling of claims where there is a non-trivial probability of catastrophic loss.
The generalized inverse-Gaussian (GIG) distribution is shown to be a natural choice for the \emph{a priori} ultimate loss distribution.
For particular GIG parameter choices, the best-estimate ultimate loss process can be written as a rational function of the paid-claims process.
We extend the model to include a second paid-claims process, and allow the two processes to be dependent.
The results obtained can be applied to the modelling of multiple lines of business or multiple origin years.
The multi-dimensional model has the property that the dimensionality of calculations remains low, regardless of the number of paid-claims processes.
An algorithm is provided for the simulation of the paid-claims processes.

\begin{center}
{\scriptsize {\bf Non-life reserving, claims development, reinsurance, best estimate of ultimate loss,\\ information-based asset pricing,  L\'evy processes, stable processes.  \\ \vspace{-0.2cm} 
} }
\end{center}
\end{abstract}
\maketitle

\vspace{-0.5cm}
\section{Introduction}
The purpose of this paper is to introduce a class of non-life insurance reserving models based on the use of a stable-$\half$ random bridge to simulate the paid-claims process. Our approach to the non-life reserving problem is to use the methods of information-based asset pricing, as represented, for example, in Brody \textit{et al.}~\cite{BHM1, BHM2, BHM3}, Hoyle \cite{H2010}, Hoyle \textit{et al.}~\cite{HHM1}, and Macrina \cite{MPhD2006}, to formulate the reserves policy in the form of a valuation problem. The term ``stable process" refers here to a strictly stable process with index $\a\in(0,2)$;
thus, we are excluding the case of Brownian motion ($\a=2$).
The use of stable processes for modelling prices in financial markets was proposed by Mandelbrot \cite{BBM1963} in his analysis of cotton futures.
The \levy densities of stable processes exhibit power-law tail decay.
As a result, the behaviour of stable processes is wild, and the trajectories have frequent large jumps.
The variance of a stable random variable is infinite.
If $\a\leq 1$, the expectation either does not exist or is infinite. 
This heavy-tailed behaviour makes stable processes ill-suited to certain applications in finance, such as forecasting and option pricing.
To overcome some of these drawbacks, so-called tempered stable processes have been introduced.
A tempered stable process is a pure-jump \levy process, and its \levy density is the exponentially dampened \levy density of a stable process.
Exponential dampening improves the integrability of the process to the extent that all the moments of a tempered stable process exist.
Tempered stable processes, however, do not possess the time-scaling properties of stable processes.

In this paper we apply stable-$\half$ random bridges to the modelling of cumulative losses.
The techniques presented are equally applicable to cumulative gains.
The integrability of a stable-$\half$ random bridge depends on the integrability of its terminal distribution.
At some fixed future time, the $n$th moment of the process is finite if and only if the $n$th moment of its terminal value is finite.
Thus a stable-$\half$ random bridge with an integrable terminal distribution can be considered to be a dampened stable-$\half$ subordinator.
In fact, the stable-$\half$ random bridge is a generalisation of the tempered stable-$\half$ subordinator.
If the \levy density of a stable-$\half$ subordinator is exponentially dampened, the result is an inverse-Gaussian (IG) process.
We shall see that the IG process is a special case of a stable-$\half$ random bridge.

The non-life reserving problem is in brief as follows. 
An insurance company incurs losses when certain events occur.
An event might be, e.g., a period of high wind, the flooding of a river, or a motor accident. 
The losses are the costs of recompensing policy-holders disadvantaged by an event. 
These costs might cover, e.g.,  repairs to property, replacement of damaged items, loss of business, medical care, and so on.
Although the loss is deemed to have been incurred by the company on the date of the event (the `loss date' or `accident date'), payment is rarely made immediately.
Delays occur because loss is not always immediately reported, because the extent of the costs takes time to emerge, because the company's obligation to pay takes time to establish, and so on.
In return for covering policy-holder risk, the company receives premiums.
The premiums received over a given period should, typically, be sufficient to cover the losses the company incurs over that period.
Since losses can take many years to pay in full, the company sets aside some of the premiums to cover future payments: these are called ``reserves".
If the reserves are set too low, the company may struggle to cover its liabilities, leading to insolvency.
Large increases in the reserves required due to a worsening in the expected future development of liabilities also cause problems.
If the reserves are set too high, shareholders or regulators may complain  that the company is withholding profits.
Thus it is important that the company should work out its ultimate liability as accurately as possible when deciding the level of reserves to set.

We use a stable-$\half$ random bridge to model the paid-claims process (i.e.~cumulative amount paid to date) of an insurance company. 
The losses contributing to the paid-claims process are assumed to have occurred in a fixed interval of time.
Sometimes claims-handling information about individual losses is known, such as that contained in police or loss-adjuster reports.
In the model we present,  the paid-claims process is regarded as providing all relevant information.
We derive the conditional distribution of the company's total liability given the paid-claims process, and then estimate recoveries from reinsurance treaties on the total liability.
The expressions arising in such estimates are similar to the expectations encountered in the pricing of call spreads on stock prices.
We examine the upper tail of the conditional distribution of the ultimate liability,
and show that it is as heavy as the \emph{a priori} tail.
This has an interesting interpretation in the case when the insurer is exposed to a catastrophic loss.
At time $t<T$, the probability of a catastrophic loss occurring in the interval $[t,T]$ decreases as $t$ approaches $T$.
However, in some sense, the \emph{size} of a catastrophic loss does not decrease as $t$ approaches $T$, since the tail of the conditional distribution of the cumulative loss does not thin.
When the \emph{a priori} total loss has a generalized inverse-Gaussian distribution, the model is particularly tractable.
We present a family of special cases where the expected total loss can be expressed as a rational function of the current value of the paid-claims process.
That is, each member of the family is a martingale that can be written as a rational function of an increasing process. 

The model can be extended to more than one paid-claims process.
We consider the case of two processes that are not independent, and  have different activity parameters. There are two ultimate losses to estimate.
We give formulae for the expected values of the ultimate losses given both paid-claims processes.
The numerical computations  are no more difficult than those of the one-dimensional case.
We demonstrate how to calculate the \emph{a priori} correlation between the ultimate liabilities.
The correlation can be used as a calibration tool when one models the cumulative losses arising from related lines of business (e.g.~personal motor and commercial motor).
We also describe how to simulate sample paths of the stable-$\half$ random bridge, and how to use a deterministic time-change to adjust the model when the paid-claims process is expected to develop non-linearly.


\section{\levy processes and stable processes} \label{sec:levy_processes}

\noindent We assume
that the reader is familiar with the theory of L\'evy processes, as discussed, 
e.g., in
\cite{Bert1996}, 
\cite{BH1}, 
\cite{CT2004},
\cite{Kyp2006}, \cite{Sato1999}, \cite {Schoutens2004}. A L\'evy process on a probability space 
$({\Omega},{\mathcal F},{\mathbb Q})$ with filtration $\{{\mathcal F}_t\}_{t\geq0}$ is a process $\{X_t\}_{t\geq0}$ such that $X_0=0$, 
$X_t-X_s$ is independent of ${\mathcal F}_s$ for $t\geq s$, and 
${\mathbb P}(X_t-X_s\leq y) = {\mathbb P}(X_{t+h}-X_{s+h}\leq y)$. 
We shall assume that $\{{\mathcal F}_t\}$ satisfies the usual conditions, and that the various processes under consideration in what follows are right-continuous with left limits.  
If a \levy process $\{S_t\}$ has the scaling property
\begin{equation}
	\label{eq:scaling}
	\{k^{-1/\a}S_{kt}\}_{t\geq 0}\law\{S_t\}_{t\geq 0} \qquad \text{for $k>0$},
\end{equation}
we call it a (strictly) stable process with index $\a$, or a stable-$\a$ process.
It can be shown that for (\ref{eq:scaling}) to hold we must have $\a\in(0,2]$.
A stable-2 process is a scaled Wiener process. A stable-1 process is a Cauchy process.
An increasing \levy process is called a subordinator. If $\{S_t\}$ is a stable-$\a$ subordinator, the Laplace transform of $S_t$ exists and is given by
\begin{equation}
	 \E[\e^{-\l S_t}]=\exp(-\k t \l^{\a}),
\end{equation}
for $\l\geq 0$, where $\k>0$ and $\a$ is further restricted to $\a\in(0,1)$.
Now suppose that $\{S_t\}$ is a \SHS.
The Laplace transform of $S_t$ is for $\l\geq 0$ given by
\begin{equation}
	\E[\e^{-\l S_t}]=\exp\left\{-\frac{ct\sqrt \l}{\sqrt 2}\right\},
\end{equation}
for some $c>0$, and $\{S_t\}$ satisfies the scaling property $\{k^{-2}S_{kt}\}\law \{S_t\}$, for $k>0$.
The random variable $S_t$ has a ``\levy distribution" with density
\begin{equation} 
	\label{eq:StableDen}
	f_t(x)=\1_{\{x>0\}} \,\frac{ct}{\sqrt{2\pi}\,x^{3/2}} \, \exp\left(-\half \frac{c^2t^2}{x}  \right).
\end{equation}
We call $c$ the ``activity parameter".
The density (\ref{eq:StableDen}) is bounded for $t>0$ and is strictly positive for $x>0$.
Integrating (\ref{eq:StableDen}) yields the distribution function
\begin{equation}
	\int_{0}^{x}f_t(y) \d y= 2\, \Phi \!\left[-ctx^{-1/2} \right],
\end{equation}
where $\Phi[x]$ is the normal distribution function.
The random variable $S_t$ has infinite mean: indeed, $\E[S_t^p]<\infty$ if and only if $p<1/2$.
The density of $1/S_t$ is
\begin{equation} 
	x\mapsto \1{\{x>0\}} \frac{ct}{\sqrt 2 \, \G[1/2]}x^{-1/2} \, \exp\left(-\half c^2t^2x  \right).
\end{equation}
Thus the increments of $\{S_t\}$ are distributed as reciprocals of gamma random variables.
Letting $\{W_t\}$ be a Wiener process, we define the so-called exceedence times $\{\tau_t\}_{t\geq 0}$ by setting
	$\tau_t=\inf\{s: W_s>ct\}$.
Then by Feller \cite{Feller2}, X.7, we have
\begin{equation}
	\{S_t\}\law \{\tau_t\}.
\end{equation}

\section{Stable-$\half$ bridges}
\noindent In what follows we make use of the theory of  \levy random bridges set out in Hoyle \textit{et al.}~\cite{HHM1}, where some of their properties are derived.
We call the random bridge of a stable-$\half$ subordinator a stable-$\half$ random bridge.
First we remark on properties of stable-$\half$ bridges, and then in the next section we deduce properties of 
stable-$\half$ random bridges.

We proceed as follows. Fix $z>0$ and let $\{S_{tT}^{(z)}\}_{0\leq t\leq T}$ be a bridge of the process $\{S_t\}$ to the value $z$ and time $T$.
Thus we are led to consider a stable-$\half$ subordinator conditioned to arrive (and terminate) at $z$ at time $T$.
Since \levy bridges are Markov processes, we have the transition law 
\begin{equation}
	\label{eq:markov}
	\Q\left[S_{tT}^{(z)} \in \dd y \left|\,  S_{sT}^{(z)} =x\right. \right]=f_{t-s,T-s}(y-x;z-x) \d y,
\end{equation}
where
\begin{align}
	f_{tT}(y;z)&=\frac{f_t(y)f_{T-t}(z-y)}{f_T(z)} \label{eq:kernelA}
		\\&=\1{\{ 0<y\leq z\}}\frac{1}{\sqrt{{2\pi}}} \frac{ct(T-t)}{T}\frac{\exp\left( -\frac{1}{2}  \frac{c^2(Ty-tz)^2}{yz(z-y)}\right)}{\left(y-y^2/z\right)^{3/2}}.
		\label{eq:kernel}
\end{align}
Note that $f_{tT}(y;z)$ is the density function of the random variable $S^{(z)}_{tT}$.
This density is bounded, and has bounded support, so 
\begin{equation}
	\E\left[\left(S^{(z)}_{tT}\right)^p\right]<\infty \qquad \text{for $p>0$.}
\end{equation}
Integration of (\ref{eq:kernel}) yields the following distribution function for $y\in[0,z]$:
\begin{equation}
F_{tT}(y;z)=
			\Phi\left[\frac{c(Ty-tz)}{\sqrt{yz(z-y)}}\right]+\left(1-\frac{2t}{T} \right)\exp \left( {2c^2t(T-t)/z} \right)\,
				\Phi\left[\frac{c((2t-T)y-tz)}{\sqrt{yz(z-y)}}\right]. \label{eq:DFIGLRB}
\end{equation}

\begin{rem}
\rm {When $t=\half T$, the second term in (\ref{eq:DFIGLRB}) vanishes.
The distribution function is analytically invertible, and we obtain the following identity, where $Z\sim N(0,1)$:}
\begin{equation}
	\label{eq:GT2}
	S^{(z)}_{T/2,T}\law \half z\left(1+\frac{Z}{\sqrt{c^2T^2/z+Z^2}} \right).
\end{equation}
.
\end{rem}

\begin{prop}
	\label{coro:SBscale}
	For fixed $k>0$, the \SH bridge $\{ S^{(z)}_{tT}\}$ satisfies the scaling property
	\begin{equation}
			\left\{ S^{(z)}_{tT} \right\}_{0\leq t \leq T}\law\left\{k^{-2} S^{(k^2z)}_{kt,kT} \right\}_{0\leq t \leq T}.
	\end{equation}
\end{prop}
\begin{proof}
	The transition probabilities of the two processes are given by
	\begin{align}
		\label{eq:conddist}
		\Q\left[S^{(z)}_{tT} \leq y \left|\, S^{(z)}_{sT}=x \right.\right]&=F_{t-s,T-s}(y-x;z-x),
		\\\Q\left[k^{-2} S^{(k^2z)}_{kt,kT} \leq y \left|\, k^{-2} S^{(k^2z)}_{ks,kT}=x \right.\right]
																																&=F_{k(t-s),k(T-s)}(k^2y-k^2x;k^2z-k^2x),
	\end{align}
	for $0\leq s <t<T$.
	It follows by use of (\ref{eq:DFIGLRB})  that these probabilities are equal.
\end{proof}

\vspace{0.1cm}
\begin{prop}
As $c\rightarrow \infty$ it holds that
	\begin{equation}
			\{S^{(z)}_{tT}\} \overset{\text{law}}{\longrightarrow}\{\tfrac{t}{T} z \}. 
	\end{equation}
\end{prop}
\begin{proof}
	Fix $z>0$.
	It is sufficient to show that
	\begin{equation}
		\lim_{c\rightarrow\infty} F_{tT}(y;z)=\1{\{Ty\geq tz\}} 
	\end{equation}
	for Lebesgue-a.e.~$y\in (0,z)$, since this is equivalent to
	\begin{equation}
		\lim_{c\rightarrow\infty} \Q\left[\left|S^{(z)}_{tT}-\tfrac{t}{T} z \right|<\varepsilon \right]=1 
	\end{equation}
	for all $t\in[0,T]$ and any $\varepsilon>0$.
	Define $\a$ by
	\begin{equation}
		\a=-\frac{(2t-T)y-tz}{\sqrt{yz(z-y)}},
	\end{equation}
	and note that $\a>0$ for $y\in(0,z)$.
	The inequality \citep[7.1.13]{AS1964} states that
	\begin{equation}
		\e^{x^2} \int_{x}^{\infty} \e^{-t^2} \d t\leq \frac{1}{x+\sqrt{x^2+4/\pi}} \qquad (x>0),
	\end{equation}
	from which we deduce
	\begin{align}
		\label{eq:ineq}
		\e^{2c^2t(T-t)/z}\Phi[-\a c]&\leq \exp \left({2c^2t(T-t)/z} \right) \sqrt{\frac{2}{\pi}} \frac{\exp \left({-\a^2c^2/2} \right)}{\a c+\sqrt{\a^2c^2+2/\pi}} \nonumber
		\\ &=\sqrt{\frac{2}{\pi}} \frac{\exp\left(-c^2\frac{(Ty-tz)^2}{2y(z-y)z}\right)}{\a c+\sqrt{\a^2c^2+2/\pi}}.
	\end{align}
	Since the left-hand side of (\ref{eq:ineq}) is positive, we see that
	\begin{equation}
		\lim_{c\rightarrow\infty} \e^{2c^2t(T-t)/z}\Phi[-\a c]=0.
	\end{equation}
	Then we have
	\begin{align}
		\lim_{c\rightarrow\infty} F_{tT}(y;z)&=\lim_{c\rightarrow\infty} \Phi\left[\frac{c(Ty-tz)}{\sqrt{yz(z-y)}}\right] 
				+\left(1-\frac{2t}{T} \right)\lim_{c\rightarrow\infty} \exp \left({2c^2t(T-t)/z} \right)\Phi[-\a c] \nonumber
				\\&= \1_{\{ Ty-tz\geq 0 \}}-\tfrac{1}{2} \1_{\{ Ty=tz \}},
	\end{align}
	which completes the proof.
\end{proof}

We define the incomplete first moment $M_{tT}(y;z)$ of $S^{(z)}_{tT}$ by
\begin{equation}
		M_{tT}(y;z)=\int_{0}^y u \, f_{tT}(u;z) \d u \qquad (0\leq y\leq z).
\end{equation}
Straightforward use of calculus gives
\begin{equation}
	M_{tT}(y;z)=\frac{t}{T}z \left\{ \Phi\left[\frac{c(Ty-tz)}{\sqrt{yz(z-y)}}\right]-\exp \left({2c^2t(T-t)/z}\right)\,
				\Phi\left[\frac{c((2t-T)y-tz)}{\sqrt{yz(z-y)}}\right] \right\}.	
\end{equation}
We can also calculate the second moment of $S^{(z)}_{tT}$.
The result is
\begin{equation}
	\E\left[ \left(S^{(z)}_{tT}\right)^2 \right]=
	\frac{t}{T}z^2\left\{1-c(T-t)\exp \left( {\frac{c^2T^2}{2z}}\right)\sqrt{\frac{2\pi}{z}}\,\Phi\left[-cTz^{-1/2} \right]  \right\}.
\end{equation}
It then follows from equation (\ref{eq:markov}) that for $0\leq s <t <T$ we have
	\begin{equation}
		\label{eq:m1}
		\E\left[S_{tT}^{(z)}\left|\, S_{sT}^{(z)}=x \right.\right]=\frac{T-t}{T-s} x+\frac{t-s}{T-s}z,
	\end{equation}
	and
	\begin{multline}
		\label{eq:m2}
		\E\left[\left.\left(S_{tT}^{(z)}\right)^2 \,\right| S_{sT}^{(z)}=x \right] 
		\\=\frac{t-s}{T-s}(z-x)^2
			\left\{1-c(T-t)\exp \left( {\frac{c^2(T-s)^2}{2(z-x)}} \right)\sqrt{\frac{2\pi}{(z-x)}}\,\Phi\left[-c\frac{T-s}{\sqrt{z-x}} \right]  \right\}.
	\end{multline}
\begin{figure}[ht]
	\begin{center}
		\subfigure[$c=1$.]{\includegraphics[scale=.9]{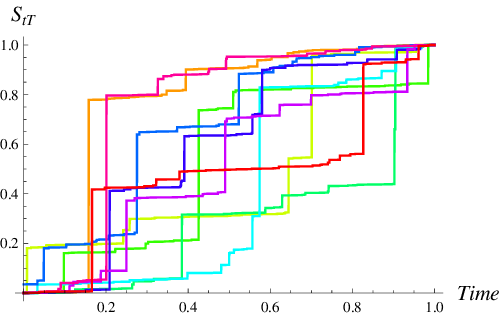}}
		\subfigure[$c=5$.]{\includegraphics[scale=.9]{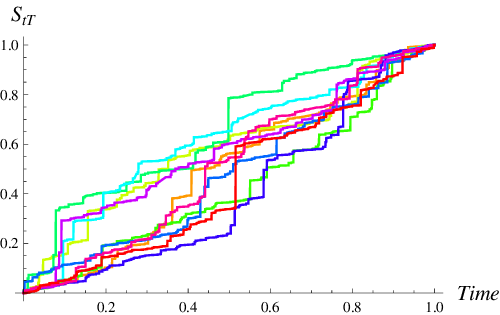}}
	\end{center}
	\caption[Stable-$\half$ bridge simulations]{%
			Simulations of the stable-$\half$ bridge demonstrating the influence of the activity parameter $c$.
			Qualitatively speaking, increasing the value of $c$ decreases the frequency of large jumps, and increases the frequency of small jumps.
			}
\end{figure}

\section{Stable-$\half$ random bridges}
\noindent Let $\nu$ be a probability law on $\R_+$. 
We say that $\{\xi_{tT}\}_{0\leq t \leq T}$ is a stable-$\half$ random bridge with terminal law $\nu$ if $\xi_{TT}$ has law $\nu$,  and 
there exists a stable-$\half$ subordinator $\{S_t\}$ such that
				\begin{equation}
				\label{item:condLevy}
							\Q\left[\xi_{t_1,T}\leq x_1,\ldots,\xi_{t_n,T} \leq x_n \left|\, \xi_{TT} = z \right.\right]
							 =\Q\left[S_{t_1}\leq x_1,\ldots,S_{t_n} \leq x_n \left|\, S_{T} = z \right.\right]
				\end{equation}
				for every $n\in\mathbb{N}_+$, every $0<t_1<\cdots<t_n<T$, every $(x_1,\ldots,x_n)\in\R^n$, and $\nu$-a.e.~$z$.
It is useful to think of $\{\xi_{tT}\}$ as a stable-$\half$ bridge to a random variable with law $\nu$.
The finite-dimensional distributions of $\{\xi_{tT}\}$ are
\begin{equation}
	 \Q\left[ \xi_{t_1,T}\in \dd x_1,\ldots, \xi_{t_n,T}\in\dd x_n, \xi_{TT}\in\dd z \right]= 
 	\prod_{i=1}^{n} \left[f_{t_i-t_{i-1}}(x_i-x_{i-1}) \d x_i\right] \psi_{t_{n}}(\dd z;x_{n}),
\end{equation}
where the (un-normalised) measure $\psi_t(\dd z;\xi)$ is given by
\begin{align}
	\psi_0(\dd z;\xi)&=\nu(\dd z),
	\\\psi_t(\dd z;\xi)&=\frac{f_{T-t}(z-\xi)}{f_T(z)}\nu(\dd z) \nonumber
	\\ &=\1_{\{z>\xi\}}\frac{\left(1-\frac{t}{T}\right)}{ \left(1-\frac{\xi}{z}\right)^{3/2}}
				\exp\left\{\half\left(\frac{c^2T^2}{z}-\frac{c^2(T-t)^2}{z-\xi} \right)  \right\}\,\nu(\dd z),
\end{align}
for $0<t<T$.
It can be shown that $\{\xi_{tT}\}$ is a Markov process with transition law
\begin{equation}
	\label{eq:LRBtranslaw}
		\begin{aligned}
			\Q[\xi_{tT} \in \dd y \,|\, \xi_{sT}=x]&=\frac{\psi_t(\R;y)}{\psi_s(\R;x)} f_{t-s}(y-x)\d y,
			\\\Q[\xi_{TT} \in \dd y \,|\, \xi_{sT}=x]&=\frac{\psi_s(\dd y;x)}{\psi_s(\R;x)},
		\end{aligned}
\end{equation}
for $0\leq s<t<T$. Now fix a time $s<T$ and define a process $\{\eta_{tT}\}_{s\leq t \leq T}$ by
\begin{equation}
	\eta_{tT}=\xi_{tT}-\xi_{sT}.
\end{equation}
Then $\{\eta_{tT}\}$	is a stable-$\half$ bridge with terminal law $f_{T-s}(x)\psi_{T-s}(\R;x) \d x$.
Furthermore, given $\xi_{sT}$, $\{\eta_{tT}\}$ is a stable-$\half$ bridge with terminal law
\begin{equation}
	\nu^*(A)=\frac{\psi_s(A+\xi_{sT};\xi_{sT})}{\psi_s(\R;\xi_{sT})},
\end{equation}
where $A+y$ is the set $\{x:x-y\in A\}$.
This is the ``dynamic consistency" property of stable-$\half$ bridges.
For the financial significance of the dynamic consistency see  \citep{BHM1}, \citep{BHM3},  \citep{HHM1}.

\section{The insurance model}
\noindent We approach the non-life insurance claims reserving problem by modelling the paid-claims process by a stable-$\half$ random bridge.  
We shall look at the problem of calculating the reserves required to cover the losses arising from a single line of business when we observe the paid-claims process. The model can also be used to describe the development of a single claim up to some fixed time $T$. Arjas \cite{EA1989} and Norberg \cite{RN1, RN2} provide detailed descriptions of the problem in a very general setting.
England \& Verrall \cite {EV2002} and W\"uthrich \& Merz \cite{WM2008, WM2013} survey some of the existing actuarial models. 
B\"uhlmann \cite {HB1970} and Mikosch \cite {TK2004} cover a number of important related topics.
The present work ties in also with that of Brody  \textit{et al.}~\cite {BHM3}, who use a gamma random bridge to model a cumulative loss or gain in the context of asset pricing.
The method we use has a flavour of the Bornhuetter-Ferguson model from actuarial mathematics \citep{BF1972} (see also \citep{EV2002}).
In implementing the Bornhuetter-Ferguson model, one begins with an \emph{a priori} estimate for the \emph{ultimate loss} 
(the total cumulative loss arising from the underwritten risks).
Periodically, this estimate is revised using a chain-ladder technique to take into account the \emph{a priori} estimate
and the development of the total paid (or reported) claims to date.
In the proposed model, we assume an \emph{a priori} distribution for the ultimate loss.
By conditioning on the development of the paid-claims process, we revise the ultimate loss distribution by use of a filtering technique.
In this way we continuously update the conditional distribution for the total loss.
This is as opposed to the deterministic Bornhuetter-Ferguson model in which only a point estimate is updated.
Knowledge of the conditional distribution allows one to calculate confidence intervals around the expected loss, and to calculate expected reinsurance recoveries.
Credibility theory uses Bayesian methods to calculate insurance premiums (see, e.g., B\"uhlmann \& Gisler \cite{BG2005}).
In a typical set-up, the premium a policyholder must pay for insurance cover is a functional of their future claims distribution.
The policyholder's future-claims distribution is parameterised by a rating factor $\Theta$, whose value is unknown, but for which we have an \emph{a priori} distribution.
Observation of the policyholder's claims history then leads to an updating of the distribution of $\Theta$, and hence an updating of their premium for future cover.
W\"uthrich \& Merz \cite{WM2008} discuss Bayesian reserving models based on credibility theory.
Credibility models update in discrete time, and the ultimate-loss distribution is updated indirectly through a ``rating'' variable $\Theta$.
This means that it is not straightforward to allow for an arbitrary \emph{a priori} ultimate-loss distribution in such models.
The main assumptions of the stable-$\half$ bridge model are:
\begin{enumerate}
	\item
		The claims arising from the line of business have run off at time $T$.
		That is, at time $T$ all claims have been settled, and the ultimate loss $U_T$ is known.
	\item
		$U_T$ has \emph{a priori} law $\nu$ such that $U_T>0$ and $\E[U_T^2]<\infty$.
	\item
		The paid-claims process $\{\xi_{tT}\}$ is a stable-$\half$ random bridge, and $\xi_{TT}=U_T$.
	\item \label{A4}
		The \emph{best estimate} of the ultimate loss is $U_{tT}=\E\left[U_T \left| \F^{\xi}_t \right.\right]$,
		where $\{\F^{\xi}_t\}$ is the natural filtration of $\{\xi_{tT}\}$.
\end{enumerate}

A few remarks can be made about these assumptions. Runoff at some fixed T is convenient for many lines of business. For example, in the case of motor insurance, T = 10 years is not unreasonable. For some liability classes the time-frames can be much longer.
The use of the natural filtration of $\{\xi_{tT}\}$ as the reserving filtration means that the paid-claims process is the only source of information
about the ultimate loss once the measure $\nu$ is set. We do not consider here the situation where one has access to information about claims that have been reported but not yet paid in full (such as case estimates). The information-based approach is well-suited for dealing with the general case, but we confine the discussion here to the simpler situation. 
The choice of measure $\Q$ with respect to which the expectation is taken is a  delicate issue that cannot be resolved in an unambiguous manner from an actuarial perspective, but admits a reasonably logical treatment from the point of view of finance theory under assumptions of liquidity. One can put the matter as follows. We can think of the reserve required at time $t \in [0, T]$ as being an estimate made at time $t$ of the amount needed to cover the claim payments between time $t$ and time $T$. The issue is that of clarifying the meaning of  the vague expression ``to cover".  It might be that we take $\Q$ to be the real-world measure, and that ``to cover" is interpreted in the sense of  ``to cover in expectation".  But this is problematic, since arguably what we would like to have is something more like ``to cover with reasonably high probability", and this then leads to questions of ``prudence", i.e., what probability figure should one aim for. In practice, insurers tend to interpret the best estimate of ultimate loss as being a real-world calculation, then adjust it upward for prudence. It is this process of ``adjustment" that requires clarification. 

From a  financial point of view, one can ask a different sort of question, which can be answered more precisely. We can ask, at any $t \in [0, T]$, for the price that would have to be paid (by the insurance company)  to enter into a contract where a third party assumes responsibility for the remaining payments made up to $T$. Then, providing that the company always has at hand reserves of that value, it will be guaranteed that the relevant claims can be met. There are two issues that arise with such an interpretation. We need to assume that there is a sufficiently liquid market for such contracts that they can be entered into at short notice without punitive transaction costs; and we need to assume that the probability of a ``major" event occurring in the short interval between the time that the company decides to enter such a contract, and the time that the company actually enters such a contract, is small. But such assumptions are in effect equivalent to the assumptions one usually needs to make in postulating the existence of a so-called pricing operator or ``pricing kernel" in financial markets. So on that basis we can argue that the reserve required at time $t$ is equal to the ``value" of the random payment stream over the interval $[t,T]$. This value is obtained by deflating the random payments by use of the pricing kernel, forming the conditional expectation given information up to $t$, and dividing the result by the pricing kernel at $t$. Use of the  pricing kernel supplies the required discounting and makes an appropriate risk adjustment or change of measure. 

It is the nature of insurance that the end-users of the products are hedgers---that the effect of buying insurance is to eliminate (or  diminish) various negative cash flows that might be encountered by the policy-holder. Thus the buyers of insurance reduce their risk, whereas the sellers of insurance increase their risk. It follows by the usual logic of finance theory that insurance should offer a negative excess rate of return (above the risk-free rate) for buyers of the product. This has the implication that the pricing measure 
$\Q$ assigns a probability to an (unwanted) event that is rather higher than it is in reality, which means that the premium required is also rather higher than it would be on the basis of real-world expectation. The excess premium is the reward that the insurance company gets for assuming the risk of the insured events. The value of the reserve required is therefore greater than the discounted real-world expectation of the future claims payments over the interval $[t,T]$. Going forward we shall assume that such principles are implicitly applied by the market in the pricing of insurance products, and we can call the resulting $\Q$ the \emph{actuarial measure}. 

Finally, we remark on the fact that insurance practitioners, when reserving, will  routinely discount data before modelling.
Discounting may adjust the data for the time-value of money or for the effects of claims inflation.
Claims inflation, and interest rates, though understood to be stochastic, often only result in a comparatively  small amount of uncertainty to the distribution of the ultimate loss, relative to the uncertainty surrounding the frequency and (discounted) sizes of insurance claims.
Furthermore, it is often for practical purposes reasonable to assume that claims inflation and interest rates are independent of claim frequency and size.
Hence, a stochastic reserving model, at least over time horizons that are not too large, may lose little from the assumption that interest rates and inflation rates are deterministic.
We make this assumption, and further assume that the paid-claims process has been appropriately discounted for the effects of interest and inflation. If longer time horizons are considered---involving decades, rather than years---then the situation with discounting requires closer scrutiny; but that will not be our concern here.


\section{Estimating the ultimate loss}
\noindent The conditional law of $U_T$ given information up to time $t$ is
\begin{align}
	\nu_t(\dd z)&=\frac{\psi_t(\dd z;\xi_{tT})}{\psi_t(\R;\xi_{tT})} \nonumber
	\\&=\frac{\1_{\{z>\xi_{tT}\}}\left( \frac{z}{z-\xi_{tT}} \right)^{3/2}
				\exp\left( -\frac{c^2}{2}\left(\frac{(T-t)^2}{z-\xi_{tT}}-\frac{T^2}{z} \right) \right) \nu(\dd z)}
					{\int_{\xi_{tT}}^{\infty} \left( \frac{u}{u-\xi_{tT}} \right)^{3/2}
				\exp\left( -\frac{c^2}{2}\left(\frac{(T-t)^2}{u-\xi_{tT}}-\frac{T^2}{u} \right) \right) \nu(\dd u)}.
\end{align}
The best-estimate ultimate loss is then
\begin{equation}
	U_{tT}=\int_{\xi_{tT}}^{\infty} z \, \nu_t(\dd z).
\end{equation}
At time $t \in [t, T]$, the total amount of claims yet to be paid is $U_T-\xi_{tT}$.
The amount  the insurance company sets aside to cover this quantity is called the \emph{reserve}.
The expectation of the total future payments is called the \emph{best-estimate reserve}, and can be expressed by
\begin{equation}
	R_{tT}=U_{tT}-\xi_{tT}.
\end{equation}
For prudence, the reserve may be greater than the best-estimate reserve.
However, for regulatory reasons it is sometimes required that the best-estimate reserve is reported.
The variance of the total future payments is the variance of the ultimate loss,
which is given by
\begin{equation}
	\var\left[U_T-\xi_{tT}\left| \F^{\xi}_t \right.\right]=\var\left[U_T\left| \F^{\xi}_t \right.\right]=\int_{\xi_{tT}}^{\infty} (z-U_{tT})^2\, \nu_t(\dd z).
\end{equation}

\section{The paid-claims process}
\noindent We shall give expressions for the first two conditional moments of the paid-claims process.
Using equations (\ref{eq:m1}) and (\ref{eq:m2}), and a straightforward conditioning argument, we have
\begin{equation}
	\label{eq:expectation}
	\E\left[\xi_{tT} \left|\, \F^{\xi}_s \right.\right]=\frac{T-t}{T-s}\xi_{sT}+\frac{t-s}{T-s} U_{sT},
\end{equation}
and
\begin{align}
	\E\!\left[\xi_{tT}^2\right]&=\frac{t}{T} \int_0^{\infty}
	z^2\left\{1-c(T-t)\e^{\frac{c^2T^2}{2z}}\sqrt{\frac{2\pi}{z}}\,\Phi\left[-cTz^{-1/2} \right]  \right\} \nu(\dd z) \nonumber
	\\ &=\frac{t}{T} \E\left[U_T^2\right]-c(T-t)\sqrt{2\pi}\int_0^{\infty}
	z^{3/2}\,\e^{\frac{c^2T^2}{2z}}\,\Phi\left[-cTz^{-1/2}\right] \nu(\dd z). \label{eq:var}
\end{align}
Equation (\ref{eq:expectation}) implies that the paid-claims development is expected to be linear.
We return to this point later.
Fix $s<T$ and define the relocated process $\{\eta_{tT}\}_{s\leq t \leq T}$ by
$
	\eta_{tT}=\xi_{tT}-\xi_{sT}.
$
The dynamic consistency property implies that, given $\xi_{sT}$, $\{\eta_{tT}\}$ is
a stable-$\half$ random bridge with marginal law of $\eta_{TT}$ being $\nu^*(A)=\nu_s(A+\xi_{sT})$.
Then we have
\begin{align}
	\E\!\left[\xi_{tT}^2\left|\, \F^{\xi}_s \right.\right] \nonumber
	&=\E\!\left[\eta_{tT}^2\left|\, \xi_{sT} \right.\right]+2\xi_{sT}\,\E\left[\eta_{tT} \left|\, \xi_{sT} \right.\right]+\xi_{sT}^2
	\\&=\frac{T-t}{T-s}\xi_{sT}^2+\frac{t-s}{T-s}\E\left[U^2_T\left|\,\xi_{sT}\right.\right] \nonumber
	\\& \quad-c(T-t)\sqrt{2\pi}\int_{\xi_{sT}}^{\infty}
					(z-\xi_{sT})^{\frac{3}{2}}\,\e^{\frac{c^2(T-s)^2}{2(z-\xi_{sT})}}\,\Phi\!\!\left[-\frac{c(T-s)}{\sqrt{z-\xi_{sT}}} \right]\! \nu_s(\dd z).
\end{align}

\section{Reinsurance}
\noindent An insurance company may buy reinsurance to protect itself against adverse claim developments. The resulting contracts or ``treaties" can be represented in the form of certain classes of derivatives. 
The \emph{stop-loss} and \emph{aggregate excess-of-loss} treaties are two types of reinsurance that cover some or all of the total amount of claims paid over a fixed threshold.
Under a stop-loss treaty, the reinsurance covers all the losses above a prespecified level.
If this level is $K$, then the reinsurance provider pays $(U_T-K)^+$ to the insurance company.
The ``aggregate $L$ excess of $K$" treaty is a capped stop-loss, and covers the layer $[K,K+L]$.
In this case the reinsurance provider pays an amount $(U_T-K)^+-(U_T-K-L)^+$.
The insurance company typically receives money from the reinsurance provider periodically. 
The amount received depends on the amount that they have paid on claims to-date.
If the insurer has the paid-claims process $\{\xi_{tT}\}$, and receives payments from a stop-loss treaty (at level $K$) on the fixed dates $t_1<t_2<\cdots <t_n=T$, then the amount received on date $t_i$ is
\begin{equation}
	\label{eq:RI_payment}
	(\xi_{t_i,T}-K)^+-(\xi_{t_{i-1},T}-K)^+.
\end{equation}
The expected value of payments such as (\ref{eq:RI_payment}) can be calculated using the following:

\begin{prop}
	At time $s<t<T$, the expected exceedence of $\xi_{tT}$ over some fixed $K>0$ is
		\begin{align}
			D_{st}&=\E\left[ (\xi_{tT}-K)^+\left|\, \F^{\xi}_s \right.\right] \nonumber
			\\&=\frac{T-t}{T-s}\xi_{sT}+\frac{t-s}{T-s}U_{sT}-K \nonumber
			\\&\qquad+\1{\{K>\xi_{sT}\}}(K-\xi_{sT})\int_{K}^{\infty}
					F_{t-s,T-s}(K-\xi_{sT};z-\xi_{sT})\, \nu_s(\dd z) \nonumber
			\\&\qquad-\1{\{K>\xi_{sT}\}}\int_{K}^{\infty}M_{t-s,T-s}(K-\xi_{sT};z-\xi_{sT}) \, \nu_s(\dd z).
		\end{align}
\end{prop}
\begin{proof}
	If $K\leq \xi_{sT}$ then
	\begin{align}
		\E\left[ (\xi_{tT}-K)^+\left|\, \F^{\xi}_s \right.\right]&=\E\left[ \xi_{tT}\left|\, \F^{\xi}_s \right.\right]-K \nonumber
		\\ &=\frac{T-t}{T-s}\xi_{sT}+\frac{t-s}{T-s}U_{sT}-K.
	\end{align}
	Thus we need only consider the case when $K>\xi_{sT}$.
	The $\F^{\xi}_s$-conditional law of $\xi_{tT}$ is
	\begin{equation}
		\Q[\xi_{tT}\in \dd y \,|\, \F^{\xi}_s]=\frac{\psi_t(\R;\xi_{tT})}{\psi_s(\R;\xi_{sT})}f_{t-s}(y-\xi_{sT}) \d y.
	\end{equation}
	Hence we have
	\begin{align}
		D_{st}&= \frac{1}{\psi_s(\R;\xi_{sT})}\int_K^{\infty}(y-K)\psi_t(\R;y)f_{t-s}(y-\xi_{sT})\d y \nonumber
		\\ &= \frac{1}{\psi_s(\R;\xi_{sT})} \nonumber
				\int_K^{\infty}(y-K)\int_K^{\infty}\frac{f_{T-t}(z-y)}{f_T(z)}\,\nu(\dd z) \, f_{t-s}(y-\xi_{sT})\d y
		\\ &= \frac{1}{\psi_s(\R;\xi_{sT})}\int_{K}^{\infty}\int_{K}^{z}(y-K)\frac{f_{T-t}(z-y)f_{t-s}(y-\xi_{sT})}{f_T(z)} \d y \,\nu(\dd z) \nonumber
		\\ &= \int_{K}^{\infty}\int_{K}^{z}(y-K)f_{t-s,T-s}(y-\xi_{sT};z-\xi_{sT}) \d y \, \nu_s(\dd z).
	\end{align}
	Making the change of variable $x=y-\xi_{sT}$ yields
	\begin{align}
		D_{st} &= \int_{K}^{\infty}\int_{K-\xi_{sT}}^{z-\xi_{sT}}(x+\xi_{sT}-K)f_{t-s,T-s}(x;z-\xi_{sT}) \d x \, \nu_s(\dd z) \nonumber
		\\ &= \int_{K}^{\infty}\left\{\frac{t-s}{T-s}(z-\xi_{sT})-M_{t-s,T-s}(K-\xi_{sT};z-\xi_{sT})\right\} \nu_s(\dd z) \nonumber
		\\ &\qquad\qquad + (\xi_{sT}-K)\int_{K}^{\infty}\left\{1-F_{t-s,T-s}(K-\xi_{sT};z-\xi_{sT})\right\} \nu_s(\dd z) \nonumber
		\\ &= \frac{T-t}{T-s}\xi_{sT}+\frac{t-s}{T-s}U_{sT}-K \nonumber
	+\int_{K}^{\infty}(K-\xi_{sT})F_{t-s,T-s}(K-\xi_{sT};z-\xi_{sT})\, \nu_s(\dd z) \nonumber
		\\&\qquad\qquad-\int_{K}^{\infty}M_{t-s,T-s}(K-\xi_{sT};z-\xi_{sT}) \, \nu_s(\dd z).
	\end{align}
\end{proof}

Suppose the insurance company has limited its liability by entering into a stop-loss reinsurance contract.
At $s\in[0,T)$, the expected reinsurance recovery between $t$ and $u$ is
\begin{equation}
	\label{eq:ReRec}
	\E\left[ (\xi_{uT}-K)^+- (\xi_{tT}-K)^+ \left|\, \F^{\xi}_s \right.\right]=D_{su}-D_{st}, 
\end{equation}
for $s<t<u\leq T$.
Using a similar method to the calculation of $D_{st}$, we can calculate the expectation of $\xi_{tT}$ conditional on it exceeding a threshold.
For a threshold $\th>\xi_{sT}$, we find
\begin{equation}
	\E[\xi_{tT} \,|\, \xi_{sT}, \xi_{tT}>\th]=
		 \frac{\frac{T-t}{T-s}\xi_{sT}+\frac{t-s}{T-s} U_{sT}-\int_{\xi_{sT}}^{\infty}M_{t-s,T-s}(\th-\xi_{sT};z-\xi_{sT}) \,\nu_s(\dd z)}
				{1-\int_{\xi_{sT}}^{\infty}F_{t-s,T-s}(\th-\xi_{sT};z-\xi_{sT}) \,\nu_s(\dd z)}.
\end{equation}
Sometimes called the conditional value-at-risk ($\mathrm{CVaR}$), this expected value is a coherent risk measure, and is a useful tool for risk management 
(McNeil  \textit{et al.} \cite{MFE2005}).
Note that $\mathrm{CVaR}$ is normally defined as an expected value conditional on a shortfall in profit.
Since we are modelling loss, and not profit, the risk we most wish to manage is on the upside.
Hence, conditioning on an exceedence is of greater interest.

\section{Tail behaviour}
\noindent In this section we consider how the probability of extreme events is affected by the paid-claims development.
Suppose that the line of business we are modelling is exposed to rare but ``catastrophic" large loss events.
In this case we assume that the \emph{a priori} distribution of the ultimate loss has a heavy right tail.
If a catastrophic loss could hit the insurance company at any time before runoff, then it is important that any conditional distributions for the ultimate loss
retain the heavy-tail property.
We shall see that in the stable-$\half$ random bridge model  the conditional distributions are as heavy-tailed as the \emph{a priori} distribution.

Assume that $U_T$ has a continuous density $p(z)$ that is positive for all $z$ above some threshold.
Then the value of $U_T$ is unbounded in the sense that
\begin{equation}
	\Q[U_T>x]>0, \quad\text{for all $x\in\R$.}
\end{equation}
Define
\begin{equation}
	\mathrm{Tail}_t=
		\lim_{L\rightarrow\infty} \frac{\Q\left[ \xi_{TT}>L \right]}{\Q\left[ \xi_{TT}-\xi_{tT}>L \left|\, \xi_{tT} \right.\right]}.
\end{equation}
If $\mathrm{Tail}_t=\infty$ then the tail of the future-payments distribution at time $t>0$ is not as heavy as the \emph{a priori} tail.
That is, a catastrophic loss at time $t$ is ``smaller" than a catastrophic loss at time 0.
If $\mathrm{Tail}_t=0$ then the tail of the future-payments distribution is greater at time $t$ than \emph{a priori}.
If $0<\mathrm{Tail}_t<\infty$ then the tail is as heavy at time $t$ as \emph{a priori}.
Using l'H\^opital's rule, we have
\begin{align}
	\mathrm{Tail}_t
			&= \lim_{L\rightarrow\infty}
			\frac{\psi_t(\R;\xi_{tT})\int_{L}^{\infty} p(z) \d z} 
			{\int_{L+\xi_{tT}}^{\infty} \left( \frac{z}{z-\xi_{tT}} \right)^{3/2}
				\exp\left( -\half c^2 \left(\frac{(T-t)^2}{z-\xi_{tT}}-\frac{T^2}{z} \right) \right) p(z) \d z} \nonumber
			\\&=  \lim_{L\rightarrow\infty}
			\frac{\psi_t(\R;\xi_{tT})\,p(L)}
			{\left( \frac{L+\xi_{tT}}{L} \right)^{3/2}
				\exp\left( -\half c^2 \left(\frac{(T-t)^2}{L}-\frac{T^2}{L+\xi_{tT}} \right) \right)p(L+\xi_{tT})} \nonumber
			\\&={\psi_t(\R;\xi_{tT})}\,\lim_{L\rightarrow\infty}\frac{p(L)}{p(L+\xi_{tT})},
\end{align}
for $t\in(0,T)$.
Some examples include:
\begin{enumerate}
	\item
		If $p(z)\propto \1_{\{z>0\}}\e^{-z}$ (exponential) then $\mathrm{Tail}_t=\psi_t(\R;\xi_{tT})\,\exp {\xi_{tT}}$.
	\item
		If $p(z)\propto \1_{\{z>0\}}\e^{-z^2}$ (half-normal) then $\mathrm{Tail}_t=\psi_t(\R;\xi_{tT})\,\exp {\xi_{tT}^2}$.
	\item
		If $p(z)\propto \1_{\{z>0\}}z^{-3/2} \e^{-1/z}$ (L\'evy) then $\mathrm{Tail}_t=\psi_t(\R;\xi_{tT})$.
\end{enumerate}
This property has an interesting parallel with the \emph{subexponential} distributions.
By definition, a random variable $X$ has a subexponential distribution if
\begin{equation}
	\lim_{L\rightarrow\infty} \frac{\Q\left[ \sum_{i=1}^n X_i>L \right]}{\Q\left[ X>L \right]} = n,
\end{equation}
where $\{X_i\}_{i=1}^n$ are independent copies of $X$ (Embrechts \textit{et al.} \cite{EKM1997}).
We note that
\begin{equation}
\lim_{L\rightarrow\infty} \frac{\Q\left[ Z_{T}>L \right]}{\Q\left[ Z_T-Z_t>L \left|\, Z_t \right.\right]}=\infty,
\end{equation}
for $\{Z_t\}$ a Brownian motion, a geometric Brownian motion, or a gamma process.
If $\{Z_t\}$ is a stable-$\half$ subordinator, so the increments of $\{Z_t\}$ are subexponential, then
\begin{equation}
\lim_{L\rightarrow\infty} \frac{\Q\left[ Z_{T}>L \right]}{\Q\left[ Z_T-Z_t>L \left|\, Z_t \right.\right]}=\frac{T}{T-t}.
\end{equation}

\section{Generalized inverse-Gaussian prior}
\noindent The three-parameter generalized inverse-Gaussian (GIG) distribution on the positive half-line has a density of the following form
(J\o rgensen \cite {Jorg1982}, Eberlein \& von Hammerstein \cite{EH2004}):
\begin{equation}
	\label{eq:GIGdensity}
	 f_{\textit{GIG}}(x;\l,\delta,\g)= 
	 \1_{\{x>0\}} \left( \frac{\g}{\delta}\right)^{\l} \frac{1}{2\,K_{\l}[\g\delta]}x^{\l-1}\exp\left(-\tfrac{1}{2}(\delta^2x^{-1}+\g^2x) \right).
\end{equation}
Here $K_{\nu}[z]$ denotes the modified Bessel function \cite{AS1964}.
The permitted parameter values are
\begin{align}
	&\delta\geq 0, && \g>0, 		&&\text{if $\l>0$,}
	\\ &\delta> 0, && \g>0, 		&&\text{if $\l=0$,}
	\\ &\delta> 0, && \g\geq 0, &&\text{if $\l<0$.}
\end{align}
If $\l>0$, the limit $\delta\rightarrow 0^+$ gives the gamma distribution.
If $\l<0$, the limit $\g\rightarrow 0^+$ yields the reciprocal-gamma distribution---this includes the \levy distribution for $\l=-\half$
(recall that the \levy distribution is the increment distribution of \SH subordinators).
The case $\l=-\half$ and $\g>0$ corresponds to the IG distribution.
If $X$ has the density (\ref{eq:GIGdensity}) then the moment $\mu_k=\E[X^k]$ is given by
\begin{align}
	\mu_k&=\frac{K_{\l+k}[\g\delta]}{K_{\l}[\g\delta]}\left(\frac{\delta}{\g} \right)^k \qquad \text{for $\l\in\R$, $\delta>0$, $\g>0$},
	\\\mu_k&=\left\{
					\begin{aligned}
						&\frac{\G[\l+k]}{\G[\l]} \left(\frac{2}{\g^2} \right)^k && k>-\l
						\\ &\infty && k\leq-\l
					\end{aligned}
				\right. \text{and $\l>0$, $\delta=0$, $\g>0$,}
	\\\mu_k&=\left\{
					\begin{aligned}
						&\frac{\G[-\l-k]}{\G[-\l]} \left(\frac{\delta^2}{2} \right)^k && k<-\l
						\\ &\infty && k\geq-\l
					\end{aligned}
				\right. \text{and $\l<0$, $\delta>0$, $\g=0$.}
\end{align}
The following identity is useful \citep[10.2.15]{AS1964}:
\begin{equation}
	\label{eq:BesselHalfInt}
	K_{ n+\frac{1}{2}}[z]= \sqrt{\tfrac{1}{2} \pi /z} \,\,\e^{-z} \sum_{j=0}^{n} (n+\tfrac{1}{2},j) (2z)^{-j}, \quad\text{for $n\in\mathbb{N}$,}
\end{equation}
where $(n + \half,j)$ denotes the Hankel symbol,
\begin{equation}
	(n + \tfrac{1} {2}, j )=\frac{ (n + j)!} {j!\, \G[n - j + 1]}.
\end{equation}
The IG process is a \levy process with increment density
\begin{equation}
		q_t(x)=\1{\{x>0\}} \frac{c t}{\sqrt{2\pi}}\frac{1}{x^{3/2}}\exp\left(-\half \frac{\g^2}{x}\left(x-\tfrac{c}{\g}t\right)^2 \right).
\end{equation}
We see that $q_t(x)=f_{\textit{GIG}}(x;-\tfrac{1}{2},ct,\g)$.
The $k$th moment of $q_t(x)$ is
\begin{equation}
	m_t^{(k)}=\sqrt{\frac{2}{\pi}}\,\g\e^{\g c t}\left(\frac{c t}{\g} \right)^{k+\frac{1}{2}} K_{k-1/2}[\g c t],
\end{equation}
for $k>0$.
Using (\ref{eq:BesselHalfInt}), we find that the first four integer moments simplify to
\begin{align}
	m_t^{(1)}&=\frac{c t}{\g},
	\\ m_t^{(2)}&=\frac{c t}{\g^3}(1+\g c t),
	\\ m_t^{(3)}&=\frac{c t}{\g^5}(3+3\g c t+\g^2 c^2 t^2),
	\\ m_t^{(4)}&=\frac{c t}{\g^7}(15+15\g c t+6 \g^2 c^2 t^2+\g^3 c^3t^3).
\end{align}

\subsection{GIG terminal distribution}
\noindent The GIG distributions constitute a natural class of \emph{a priori} distributions for the ultimate loss.
With $\g>0$ and $c>0$ fixed, we examine some properties of a paid-claims process $\{\xi_{tT}\}$ with time-$T$ density $f_{\textit{GIG}}(z;\l,cT,\g)$.
The transition law is
\begin{align}
		\Q[\xi_{tT} \in \dd y \,|\, \xi_{sT}=x]&=\frac{\psi_t(\R;y)}{\psi_s(\R;x)} f_{t-s}(y-x)\d y,
		\\\Q[\xi_{TT} \in \dd y \,|\, \xi_{sT}=x]&=\frac{\psi_s (\dd y;x)}{\psi_s(\R;x)},
\end{align}
where
\begin{align}
	\psi_0(\dd z;\xi)&=f_{\textit{GIG}}(z;\l, c T,\g) \d z,
	\\ \psi_t(\dd z;\xi)&=(1-\tfrac{t}{T}) \1_{\{z>\xi\}} 
				\frac{\exp\left({-\half c^2\left(\frac{(T-t)^2}{z-\xi}-\frac{T^2}{z} \right)}\right)}{({1-\xi/z})^{3/2}} 
				\,f_{\textit{GIG}}(z;\l,c T,\g) \d z.
\end{align}
Writing
\begin{equation}
	\kappa=\left(\frac{\g}{c T}\right)^{\l} \frac{1}{2\, K_{\l}[\g\sqrt T ]},
\end{equation}
we have
\begin{align}
	\psi_t(\R;y)&=\kappa(1-\tfrac{t}{T})\e^{-\frac{1}{2}\g^2 y} \int_{y}^{\infty}z^{\l+\frac{1}{2}}\,
			\frac{\e^{-\half c^2\frac{(T-t)^2}{z-y}-\frac{1}{2}\g^2(z-y)}}{(z-y)^{3/2}} \d z \nonumber
	\\ &=\kappa(1-\tfrac{t}{T})\e^{-\frac{1}{2}\g^2 y} \int_{0}^{\infty}(z+y)^{\l+\frac{1}{2}}\,
			\frac{\e^{-\half c^2 \frac{(T-t)^2}{z}-\frac{1}{2}\g^2z}}{z^{3/2}} \d z \nonumber
	\\ &=\frac{\kappa\sqrt{2\pi}}{cT}\e^{-\frac{1}{2}\g^2 y- \g c(T-t)} \int_{0}^{\infty}(z+y)^{\l+\frac{1}{2}}
			q_{T-t}(z) \d z.
\end{align}
Given $\xi_{tT}=y$, the best-estimate ultimate loss is
\begin{align}
	U_{tT}=\psi_t(\R;y)^{-1}\int_{y}^{\infty} z \, \psi_t(\dd z;y)
	 =\frac{\int_0^{\infty}(z+y)^{\l+\frac{3}{2}} q_{T-t}(z) \d z}{\int_0^{\infty}(z+y)^{\l+\half} q_{T-t}(z) \d z}.
\end{align}

\subsection{The case $\lambda =-\half$}
\noindent When $\l=-1/2$ we have
\begin{align}
	\frac{\psi_t(\R;y)}{\psi_s(\R;x)} f_{t-s}(y-x) 
		&=\1_{\{y-x>0\}}\frac{1}{\sqrt{2\pi}}\frac{c(t-s)}{(y-x)^{3/2}}
		\exp\left(-\frac{\g^2}{2}\frac{\left((y-x)- c(t-s)/\g \right)^2}{y-x} \right) \nonumber
	\\	& =q_{t-s}(y-x).
\end{align}
Thus $\{\xi_{tT}\}$ is an IG process.
Note that in this case $\{\xi_{tT}\}$ has independent increments.

\subsection{The case $\l=n-\frac{1}{2}$}
\noindent Here we consider the case where $\l=n-\half$, for $n\in\mathbb{N_+}$.
For convenience we write
\begin{equation}
	q^{(k)}_t(x)=f_{\textit{GIG}}(x;k-1/2,ct,\g).
\end{equation}
Hence one has $q^{(0)}_t(x)=q_t(x)$.
The transition density of $\{\xi_{tT}\}$ is then
\begin{align}
	\frac{\psi_t(\R;y)}{\psi_s(\R;x)} f_{t-s}(y-x) 
	 &=q_{t-s}(y-x)\frac{\int_{0}^{\infty}(z+y)^n q_{T-t}(z) \d z}{\int_{0}^{\infty}(z+x)^n q_{T-s}(z) \d z} \nonumber
	\\	& =q_{t-s}(y-x)\frac{\sum_{k=0}^n\binom{n}{k} m_{T-t}^{(n-k)} \,y^{k}}{\sum_{k=0}^n\binom{n}{k} m_{T-s}^{(n-k)} \,x^{k}}.
\end{align}
When $n=1$ this is
\begin{align}
	\frac{\psi_t(\R;y)}{\psi_s(\R;x)} f_{t-s}(y-x)&=q_{t-s}(y-x)\frac{y+\frac{c}{\g} (T-t)}{x+\frac{c}{\g} (T-s)} \nonumber
				\\&= \left(1- \frac{c (t-s)}{\g x+ c (T-s) }\right) q_{t-s}^{(0)}(y-x) \nonumber
				\\&\qquad\qquad+\left( \frac{ c (t-s)}{\g x+ c (T-s)} \right) q_{t-s}^{(1)}(y-x).
\end{align}
Thus the increment density is a weighted sum of GIG densities.
We shall derive a weighted sum representation for general $n$.
We can write
\begin{align}
	\int_{0}^{\infty}(z+y)^n \,q_{T-t}(z) \d z
	&=\int_{0}^{\infty}((z+x)+(y-x))^n\, q_{T-t}(z) \d z \nonumber
	\\&=\sum_{k=0}^n \binom{n}{k} (y-x)^{n-k} \int_{0}^{\infty}(z+x)^{k} \,q_{T-t}(z) \d z \nonumber
	\\&=\sum_{k=0}^n \binom{n}{k} (y-x)^{n-k} \sum_{j=0}^k \binom{k}{j}m^{(k-j)}_{T-t}\, x^{k}.
\end{align}
Then we have
\begin{align}
	\frac{\psi_t(\R;y)}{\psi_s(\R;x)} f_{t-s}(y-x) 
	 &=q_{t-s}(y-x)\frac{\int_{0}^{\infty}(z+y)^n q_{T-t}(z) \d z}{\int_{0}^{\infty}(z+x)^n q_{T-s}(z) \d z} \nonumber
	\\  &=q_{t-s}(y-x)\frac{\sum_{k=0}^n \binom{n}{k} (y-x)^{n-k} 
					\sum_{j=0}^k \binom{k}{j}m^{(k-j)}_{T-t}\, x^{j}}{\sum_{k=0}^n\binom{n}{k} m^{(n-k)}_{T-s} \,x^{k}}.\label{eq:IGSB_inc_den}
\end{align}
However, when $k\in\mathbb{N}_0$,
\begin{align}
	\frac{z^{k}\,q_{t-s}(z)}{q_{t-s}^{(k)}(z)}&=\frac{z^{k}\,f_{\textit{GIG}}(z;-1/2, c (t-s),\g)}{f_{\textit{GIG}}(z;k-1/2,c (t-s),\g)} \nonumber
		\\ &=\left( \frac{ c (t-s)}{\g} \right)^{k} \frac{K_{k-1/2}[\g c (t-s)]}{K_{1/2}[\g c (t-s)]} \nonumber
		\\ &=m^{(k)}_{t-s}. 
\end{align}
Thus we have

\begin{equation}
	\label{eq:den_id}
	(y-x)^{n-k}\, q_{t-s}(y-x)=m^{(n-k)}_{t-s} \, q_{t-s}^{(n-k)}(y-x).
\end{equation}
\\
By use of the identity (\ref{eq:den_id}), (\ref{eq:IGSB_inc_den}) can be expanded to give
\begin{equation}
	\frac{\psi_t(\R;y)}{\psi_s(\R;x)} f_{t-s}(y-x) =\sum_{k=0}^n w^{(k)}_{st}(x) \, q_{t-s}^{(k)}(y-x),
\end{equation}
where
\begin{equation}
	w^{(k)}_{st}(x)=\frac{\binom{n}{k} m^{(n-k)}_{t-s} \sum_{j=0}^k \binom{k}{j}m^{(k-j)}_{T-t}\, x^{j}}
		{\sum_{j=0}^n\binom{n}{j} m^{(n-j)}_{T-s} \,x^{j}}.
\end{equation}
Note that $w^{(k)}_{st}(x)$ is a rational function. The denominator is a polynomial of order $n$. The numerator is a polynomial of order $k\leq n$.
The transition probabilities of $\{\xi_{tT}\}$ depend on the first $n$ integer powers of the current value.
The conditional law of the ultimate loss is

\begin{equation}
	\frac{\psi_s(\dd y;\xi_{sT})}{\psi_s(\R;\xi_{sT})}=\frac{y^n q^{(0)}_{T-s}(y-\xi_{sT})}{\sum_{k=0}^n \xi_{sT}^k \, m_{T-s}^{(n-k)}} \d y.
\end{equation}
\\
One can verify that $\sum_{k=0}^n w^{(k)}_{st}(x)=1$ using the fact that IG densities are closed under convolution.
We have
\begin{equation}
	q_{T-s}(z)=\int_{0}^{z} q_{T-t}(y) q_{t-s}(z-y) \d y, \quad \text{for $0\leq s<t<T$}.
\end{equation}
For fixed $n\in\mathbb{N}_+$, we then have
\begin{align}
	\sum_{k=0}^n\binom{n}{k} m^{(n-k)}_{T-s} \,x^{k}&=\int_0^{\infty} (z+x)^n q_{T-s}(z) \d z \nonumber
	\\ &= \int_0^{\infty} (z+x)^n \int_{0}^{z} q_{T-t}(y) q_{t-s}(z-y) \d y \d z \nonumber
	\\ &= \int_{0}^{\infty} q_{T-t}(y) \int_y^{\infty} (z+x)^n \, q_{t-s}(z-y)  \d z \d y \nonumber
	\\ &= \int_{0}^{\infty} q_{T-t}(y) \int_0^{\infty} (z+y+x)^n \, q_{t-s}(z)  \d z \d y \nonumber
	\\ &= \int_{0}^{\infty} q_{t-s}(y) \left[ \sum_{k=0}^n \binom{n}{k} m^{(n-k)}_{t-s} (y+x)^{k} \right] \d y \nonumber
	\\ &= \sum_{k=0}^n \binom{n}{k} m^{(n-k)}_{t-s} \int_{0}^{\infty} (y+x)^{k} \,q_{t-s}(y) \d y \nonumber
	\\ &= \sum_{k=0}^n \binom{n}{k} m^{(n-k)}_{t-s} \sum_{j=0}^{k} \binom{k}{j} m^{(k-j)}_{T-t} x^{j},
\end{align}
which gives
\begin{equation}
	\sum_{k=0}^n\frac{ \binom{n}{k} m^{(n-k)}_{t-s} \sum_{j=0}^{k} \binom{k}{j} m^{(k-j)}_{T-t} x^{j}}
		{\sum_{j=0}^n\binom{n}{j} m^{(n-j)}_{T-s} \,x^{j}}=1.
\end{equation}

\subsection{Moments of the paid-claims process}
\noindent The best-estimate ultimate loss simplifies to
\begin{equation}
	U_{tT}=\frac{\sum_{k=0}^{n+1}\binom{n+1}{k}m^{(n+1-k)}_{T-t} \xi_{tT}^k}{\sum_{k=0}^{n}\binom{n}{k}m^{(n-k)}_{T-t} \xi_{tT}^k}.
\end{equation}
For example, when $n=1$ we obtain
\begin{equation}
	U_{tT}=\frac{c(T-t)(1+\g c(T-t))+2\g^2 c(T-t)\xi_{tT}+\g^3 \xi_{tT}^2}{\g^2c(T-t)+\g^3\xi_{tT}}.
\end{equation}
By similar calculations, we have
\begin{equation}
	\E[\xi_{TT}^m\,|\,\xi_{tT}]=\frac{\sum_{k=0}^{n+m}\binom{n+m}{k}m^{(n+m-k)}_{T-t} \xi_{tT}^k}{\sum_{k=0}^{n}\binom{n}{k}m^{(n-k)}_{T-t} \xi_{tT}^k}
	\quad \text{for $m\in\mathbb{N_+}$,}
\end{equation}
and
\begin{equation}
	\E\left[\left.\e^{\half\a^2\xi_{TT}} \,\right| \xi_{tT}\right]=
	\\	\frac{\sum_{k=0}^{n}\binom{n}{k}\bar{m}^{(n-k)}_{T-t} \xi_{tT}^k}{\sum_{k=0}^{n}\binom{n}{k}m^{(n-k)}_{T-t} \xi_{tT}^k}
					\, \exp\left(\tfrac{1}{2} \a^2\xi_{tT}-(T-t)(\bar{\g}-\g)\right),
\end{equation}
for $0<\a<\g$, where $\bar{\g}=\sqrt{\g^2-\a^2}$, and $\bar{m}^{(k)}_{t}$ is the $k$th moment of the IG distribution with parameters $\delta=ct$ and $\g=\bar{\g}$.

\section{Exposure adjustment}
\noindent We have seen that 
\begin{equation}
	\E[\xi_{tT}]=\frac{t}{T}\E[U_T].
\end{equation} 
Thus in the model so far the development of the paid-claims process is expected to be linear.
This is not always the case in practice.
In some situations the marginal exposure is (strictly) decreasing as the development approaches runoff.
This manifests itself in the form
\begin{equation}
	\frac{\partial^2}{\partial t^2} \E[\xi_{tT}]<0,
\end{equation}
for $t$ close to $T$.
A straightforward method to adjust the development pattern is through a time change.
We describe the marginal exposure of the insurer through time by a deterministic function $\varepsilon:[0,T]\rightarrow\R_+$.
The total exposure of the insurer is
\begin{equation}
	\int_{0}^T \varepsilon(s) \d s.
\end{equation}
We define the increasing function $\tau(t)$ by
\begin{equation}
	\tau(t)=T\frac{\int_0^t \varepsilon(s) \d s}{\int_0^T \varepsilon(s) \d s}.
\end{equation}
By construction $\tau(0)=0$ and $\tau(T)=T$.
Now let $\tau(t)$ determine the \emph{operational time} in the model.
We define the time-changed paid-claims process $\{\xi^{\t}_{tT}\}$ by
\begin{equation}
	\xi^{\t}_{tT}=\xi(\t(t),T),
\end{equation}
and set the reserving filtration to be the natural filtration of $\{\xi_{tT}^\t\}$.
Then we have
\begin{align}
	\E[\xi^{\t}_{tT}]=\frac{\int_0^t \varepsilon(s)\d s}{\int_0^T \varepsilon(s)\d s} \E[U_{T}]
	\\\intertext{and}
	\frac{\partial^2}{\partial t^2} \E[\xi_{tT}^{\t}]=\frac{\E[U_T]}{\int_0^T\varepsilon(s) \d s}\varepsilon'(t).
\end{align}

\begin{figure}[ht]
	\begin{center}
		\includegraphics[scale=1]{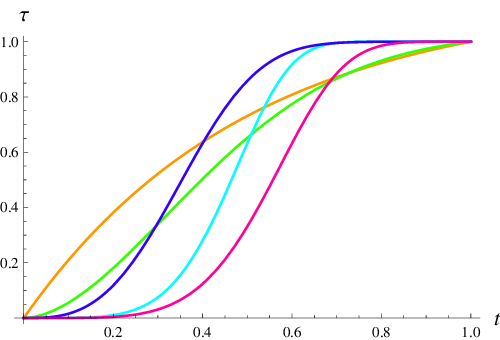}
	\end{center}
	\caption[Craighead curves: truncated Weibull time-changes.]{%
		Plots of the truncated Weibull time change for various parameters, and with $T=1$.
		The expected paid-claims development of the model will have the same profile as $\tau(t)$ (scaled by $\E[U_T]$).
		Hence, under one of the above time changes, when $t$ is close to $T$ the marginal exposure falls, 
		i.e.~$\frac{\partial^2}{\partial t^2} \E[\xi_{tT}^{\tau}]<0$\,.
}
\label{fig:Weibull}
\end{figure}

Craighead \cite{Craig1979} proposed fitting a Weibull distribution function to the development pattern of paid claims for forecasting the ultimate loss (see also Benjamin \& Eagles \cite{BE1997}).
In actuarial work, the Weibull distribution function is often called  the Craighead curve.
To achieve a similar development pattern we can use the Weibull density as the marginal exposure:
\begin{equation}
	\varepsilon(t)=\frac{b}{a}\left(t/a\right)^{b-1} \e^{-(t/a)^b}, \quad a,b>0.
\end{equation}
Then the time change $\tau(t)$ is the renormalised, truncated Weibull distribution function
\begin{equation}
	\tau(t)=T \frac{1-\e^{-(t/a)^b}}{1-\e^{-(T/a)^b}}.
\end{equation}
See Figure \ref{fig:Weibull} for plots of this function.
When $b\leq 1$, $\t'(t)$  is decreasing.
Under such a time change, the marginal exposure is decreasing for all $t\in[0,T]$.
When $b>1$, $\t'(t)$ achieves its maximum at
\begin{equation}
	t^*=a\left(\frac{b-1}{b}\right)^{1/b},
\end{equation}
and $\t'(t)$ is decreasing for $t\geq t^*$.
Thus, if $T>t^*$ then the marginal exposure is decreasing for $t\in[t^*,T]$.
If $T\leq t^*$ then the marginal exposure is increasing for $t\in[0,T]$.

\section{Simulation} \label{sec:Stable_Sim}
\noindent We consider the simulation of sample paths of a stable-$\half$ random bridge.
First, we can generalise (\ref{eq:GT2}) to
\begin{multline}
	\label{eq:simalgo}
	\left[\left. \xi((s+t)/2,T)\,\right| \xi(s,T)=y, \xi(t,T)=z \right]\law 
	 \\y+\half (z-y)\left(1+\frac{Z}{\sqrt{c^2(t-s)^2/(z-y)+Z^2}} \right),
\end{multline}
where $0< s<t\leq T$, and $Z\sim N(0,1)$.

One can then generate a discretised sample path of the form $\{\hat{\xi}(t_i,T)\}_{i=0}^{2^n}$, where $t_i=iT2^{-n}$, by a recursive algorithm of the following form:
	(1) Generate the variate $\hat{\xi}(T,T)$ with law $\nu$, and set $\hat{\xi}(0,T)=0$.
	(2) Generate $\hat{\xi}(\tfrac{1}{2}T,T)$ from $\hat{\xi}(0,T)$ and $\hat{\xi}(T,T)$ by use of the identity (\ref{eq:simalgo}).
	(3) Generate $\hat{\xi}(\tfrac{1}{4}T,T)$ from $\hat{\xi}(0,T)$ and $\hat{\xi}(\tfrac{1}{2}T,T)$,
				and then generate $\hat{\xi}(\tfrac{3}{4}T,T)$ from $\hat{\xi}(\tfrac{1}{2}T,T)$ and $\hat{\xi}(T,T)$.
	(4) Generate $\hat{\xi}(\tfrac{1}{8}T,T)$, $\hat{\xi}(\tfrac{3}{8}T,T)$, $\hat{\xi}(\tfrac{5}{8}T,T)$, $\hat{\xi}(\tfrac{7}{8}T,T)$.
	(5) Iterate.
See Figure \ref{fig:srb} for examples of  simulations.
\begin{figure}[ht]
	\begin{center}
		\subfigure[$c=3$]{\includegraphics[scale=.9]{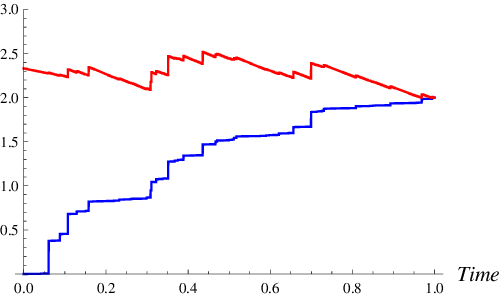}}
		\subfigure[$c=5$]{\includegraphics[scale=.9]{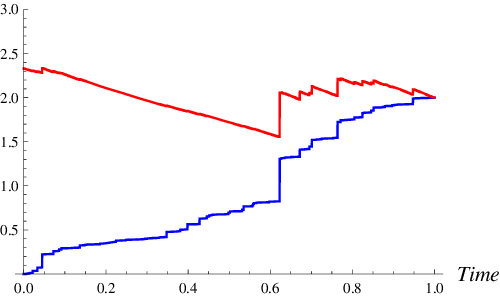}}
		\\
		\subfigure[$c=7$]{\includegraphics[scale=.9]{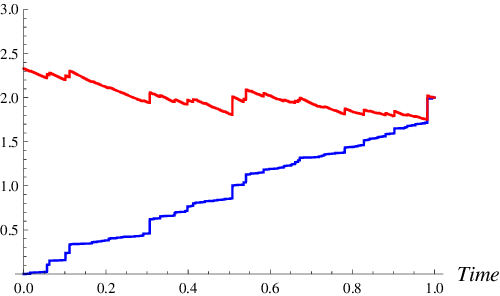}}
		\subfigure[$c=10$]{\includegraphics[scale=.9]{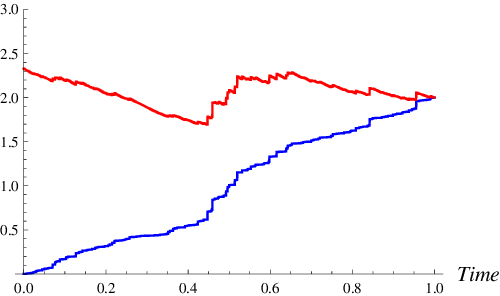}}
	\end{center}
	\caption[Simulations from the stable-$\half$ random bridge reserving model]{%
					Simulations of the paid-claims process $\{\xi_{tT}\}$ (bottom line) and the best-estimate process $\{U_{tT}\}$ (top line). 
					Various values of the activity parameter $c$ are used.
					\emph{A priori}, the ultimate loss $U_T$ has a generalized Pareto distribution (GPD) with density 
					\mbox{$f_{\textit{GPD}}(x)=\1_{\{x>1\}}\left(1+\frac{x-1}{4} \right)^{-5}.$}
					This is the GPD with scale parameter $\s=1$, location parameter $\mu=1$, and shape parameter $\xi=1/4$.}
		\label{fig:srb}
\end{figure}

\section{Multiple lines of business}
\noindent We shall generalise the paid-claims model to achieve two goals:
the first is to allow more than one paid-claims process, and allow dependence between the processes;
the second is to keep the dimensionality of the calculations low with a view to practicality.
The following results can be applied to the modelling of multiple lines of business or multiple origin years when there is dependence between loss processes.
We proceed to consider an example with two paid-claims processes. 
We set $f_t^c(x)=f_t(x)$ as given by (\ref{eq:StableDen}), and $f^c_{tT}(x)=f_{tT}(x)$ as given by (\ref{eq:kernel}). Here we have introduced the superscript to emphasise the dependence on $c$.
Let $\{S(t,T^*)\}$ be a stable-$\half$ random bridge with terminal density $p(z)=\nu(\dd z)/\dd z$, and with activity parameter $c$.
Fix a time $T<T^*$, and define a pair of paid-claims processes by
\begin{align}
	\xi^{(1)}_{tT}&=S(t,T^*) && (0\leq t\leq T),
	\\\xi^{(2)}_{tT}&=k^{2}S(\l t+T,T^*)-k^2 S(T,T^*) && (0\leq t\leq T),
\end{align}
where $\l=T^*/T-1$, and $k=d/(c\l)$ for some $d>0$.
The density of $\xi^{(1)}_{TT}$ is given by
\begin{align}
	p^{(1)}(x)&=f_T^c(x)\int_{0}^{\infty}\frac{f_{T^*-T}^c(z-x)}{f_{T^*}^c(z)}p(z) \d z \nonumber
				\\ &=\int_{0}^{\infty}f_{T,T^*}^c(x;z) \, p(z) \d z,
	\\\intertext{and the density of $\xi^{(2)}_{TT}$ is}
	p^{(2)}(x)&=k^{-2}f_{T^*-T}^c(k^{-2}x)\int_{0}^{\infty}\frac{f_{T}^c(z-k^{-2}x)}{f_{T^*}^c(z)}p(z) \d z \nonumber
				\\ &=k^{-4}f_{T^*-T}^c(k^{-2}x)\int_{0}^{\infty}\frac{f_{T}^c(k^{-2}z-k^{-2}x)}{f_{T^*}^c(k^{-2}z)}p(k^{-2}z) \d z 
				\label{eq:denI}
				\\ &=k^{-4}\int_{0}^{\infty}f_{T^*-T,T^*}^c(k^{-2}x;k^{-2}z)\,p(k^{-2}z) \d z 
				\label{eq:denII}
				\\ &=k^{-2}\int_{0}^{\infty}f_{T,\l^{-1}T^*}^{d}(x;z)\,p(k^{-2}z) \d z.
				\label{eq:denIV}
\end{align}
Here (\ref{eq:denI}) follows after a change of variable, (\ref{eq:denII}) follows from the definition of $f_{tT}(y;z)$ given in (\ref{eq:kernelA}), and (\ref{eq:denIV}) follows from the functional form of $f_{tT}(y;z)$ given in (\ref{eq:kernel}).
It follows from the dynamic consistency property that $\{\xi^{(1)}_{tT}\}$ is a stable-$\half$ random bridge with terminal density $p^{(1)}(z)$ and activity parameter $c$.
Using the dynamic consistency and scaling properties of stable-$\half$ bridges, one can show that $\{\xi^{(2)}_{tT}\}$ is a stable-$\half$ bridge with   terminal density $p^{(2)}(z)$ and activity parameter $d$.
The conditional joint density of $(\xi^{(1)}_{tT},k^{-2}\xi^{(2)}_{tT})$ is
\begin{multline}
	\Q\left[ \xi^{(1)}_{tT}\in\dd y_1, k^{-2} \xi^{(2)}_{tT}\in\dd y_2 \left|\,\xi^{(1)}_{sT}=x_1, k^{-2}\xi^{(2)}_{sT}=x_2 \right.\right]=
	\\ \left\{
	\int_{z=x_1+x_2}^{\infty} \frac{f_{T^*-(1+\l)t}^c(z-(y_1+y_2))}{f^c_{T^*-(1+\l)s}(z-(x_1+x_2))} p(z) \d z \right\}
	f_{t-s}^c(y_1-x_1) \d y_1\, f_{\l (t-s)}^c(y_2-x_2) \d y_2,
\end{multline}
for $0\leq s<t\leq T$.
Then we have
\begin{align}
	&\Q\left[ \xi^{(1)}_{tT}+k^{-2} \xi^{(2)}_{tT}\in\dd y \left|\,\xi^{(1)}_{sT}=x_1, k^{-2}\xi^{(2)}_{sT}=x_2 \right.\right] \nonumber
	\\ &\quad=\left\{ 
		\int_{z=x_1+x_2}^{\infty} \frac{f_{T^*-(1+\l)t}^c(z-y)}{f^c_{T^*-(1+\l)s}(z-(x_1+x_2))} p(z) \d z \right\} 
		f_{(1+\l) (t-s)}^c(y-(x_1+x_2)) \d y \nonumber
	\\ &\quad=\left\{ \int_{z=x_1+x_2}^{\infty} f_{(1+\l)(t-s),T^*-(1+\l)s}^c(z-(x_1+x_2);y-(x_1+x_2)) \, p(z) \d z \right\} \dd y;
\end{align}
and, given $\xi^{(1)}_{sT}=x_1$ and $k^{-2}\xi^{(2)}_{sT}=x_2$, the marginal density of $\xi^{(1)}_{tT}$ is
\begin{align}
	y_1&\mapsto\int_{z=x_1+x_2}^{\infty}f^c_{t-s,T^*-(1+\l)s}(y_1-x_1;z-(x_1+x_2)) \, p(z) \d z,
	\\\intertext{and the marginal density of $k^{-2}\xi^{(2)}_{tT}$ is}
	y_2&\mapsto\int_{z=x_1+x_2}^{\infty}f^c_{\l(t-s),T^*-(1+\l)s}(y_2-x_2;z-(x_1+x_2)) \, p(z) \d z.
\end{align}

\section{Correlation}
The \emph{a priori} correlation between the terminal values is well defined when the second moment of $\nu$ is finite.
The correlation can be used as a tool in the calibration of the model.
Assuming that $\E[S(T^*,T^*)^2]<\infty$, the correlation is defined as
\begin{equation}
	\frac{\E\left[\xi_{TT}^{(1)} \, \xi^{(2)}_{TT} \right]-\E\left[\xi^{(1)}_{TT}\right]\E\left[\xi^{(2)}_{TT}\right]}
	{\sqrt{\left(\E\left[ \left(\xi^{(1)}_{TT}\right)^2 \right]-\E\left[\xi^{(1)}_{TT}\right]^2\right)
	\left(\E\left[ \left(\xi^{(2)}_{TT}\right)^2 \right]-\E\left[\xi^{(2)}_{TT}\right]^2\right)}}.
	\label{eq:correlation}
\end{equation}
We shall calculate each of the components of (\ref{eq:correlation}) separately.
First, we obtain
\begin{equation}
	\label{eq:corr_E1}
	\E\left[\xi^{(1)}_{TT}\right]=\E[S(T,T^*)]=\frac{T}{T^*}\E[S(T^*,T^*)].
\end{equation}
Noting that 
\begin{align}
	\xi^{(2)}_{TT}&=k^2(S(T^*,T^*)-S(T,T^*)) \nonumber
	\\&\law k^2 S(T^*-T,T^*),
\end{align}
we have
\begin{align}
	\E\left[\xi^{(2)}_{TT}\right]&=k^2 \E[S(T^*-T,T^*)] \nonumber
	\\&=k^2\left(1-\frac{T}{T^*}\right)\E[S(T^*,T^*)]. \label{eq:corr_E2}
\end{align}
The second moments of $\xi^{(1)}_{TT}$ and $\xi^{(2)}_{TT}$ follow from (\ref{eq:var}), and are given by
\begin{align}
	\E\left[\left(\xi^{(1)}_{TT}\right)^2\right]&=	\frac{T}{T^*} \E\left[S(T^*,T^*)^2\right]-(T^*-T)C_{T^*}, \label{eq:mom1}
	\\\intertext{and}
	\E\left[\left(\xi^{(2)}_{TT}\right)^2\right]&=k^4 \left(1-\frac{T}{T^*}\right) \E\left[S(T^*,T^*)^2\right]-k^4TC_{T^*},\label{eq:mom2}
\end{align}
where
\begin{equation}
	C_{T^*}=c\sqrt{2\pi}\int_0^{\infty}
						z^{3/2}\,\e^{\frac{c^2T^{*2}}{2z}}\,\Phi\left[-cT^*z^{-1/2}\right] \, p(z) \d z.
\end{equation}
The final term required for working out the correlation is the cross moment.
This is
\begin{align}
	\E\left[\xi^{(1)}_{TT} \, \xi^{(2)}_{TT} \right]
		&= k^2\E\left[S(T,T^*)\left(S(T^*,T^*)-S(T,T^*) \right) \right] \nonumber
		\\ &=k^2\E\left[S(T,T^*)S(T^*,T^*)\right]-k^2\E\left[S(T,T^*)^2 \right]. \label{eq:corr_eq}
\end{align}
The first term on the right of (\ref{eq:corr_eq}) is 
\begin{align}
	k^2\E\left[S(T,T^*)S(T^*,T^*)\right]
		&=k^2\int_0^{\infty}\int_0^{\infty}x\, y \, \frac{f_T^c(x)f_{T^*-T}^c(y-x)}{f_{T^*}^c(y)} \d x \, p(y) \d y \nonumber
		\\ &=k^2\int_0^{\infty}\int_0^{\infty}x \, y\, f_{T,T^*}^c(x;y) \d x \, p(y) \d y \nonumber
		\\ &=k^2\frac{T}{T^*}\int_0^{\infty}y^2 \, p(y) \d y \nonumber
		\\ &=k^2\frac{T}{T^*}\E[S(T^*,T^*)^2].
\end{align}
The second term on the right of (\ref{eq:corr_eq}) is given by (\ref{eq:mom1}).
Hence we have
\begin{equation}
	\label{eq:corr_x}
	\E\left[\xi^{(1)}_{TT} \, \xi^{(2)}_{TT} \right]=k^2(T^*-T)C_{T^*}.
\end{equation}
The expression for the correlation follows from equations (\ref{eq:corr_E1}), (\ref{eq:corr_E2}), (\ref{eq:mom1}), (\ref{eq:mom2}), (\ref{eq:corr_x}).

\section{Ultimate loss estimation}
\noindent In conclusion we estimate the terminal values of the paid-claims processes.
At time $t<T$, the best-estimate ultimate loss of $\{\xi_{tT}^{(1)}\}$ (or, indeed, $\{\xi_{tT}^{(2)}\}$) depends on the two values $\xi_{tT}^{(1)}$ and $\xi_{tT}^{(2)}$.
The best-estimate ultimate loss of $\{\xi_{tT}^{(1)}\}$ is
\begin{align}
	U_{tT}^{(1)}&=\E\left[\xi^{(1)}_{TT} \left|\, \xi^{(1)}_{tT}=x_1, \xi^{(2)}_{tT}=x_2 \right.\right] \nonumber
	\\ &=\E\left[S(T,T^*) \left|\, S(t,T^*)=x_1, S(T+\l t,T^*)-S(T,T^*)=k^{-2}x_2 \right.\right] \nonumber
	\\ &=\E\left[S(T+\l t,T^*)\left|\, S(t,T^*)=x_1, S(T+\l t,T^*)-S(T,T^*)=k^{-2}x_2 \right.\right]-k^{-2}x_2 \nonumber
	\\ &=\E\left[S(T+\l t,T^*) \left|\, S(t,T^*)=x_1, S((1+\l)t,T^*)-S(t,T^*)=k^{-2}x_2 \right.\right]-k^{-2}x_2 \label{eq:usereorder}
	\\ &=\E\left[S(T+\l t,T^*) \left|\, S((1+\l)t,T^*)=x_1+k^{-2}x_2 \right.\right]-k^{-2}x_2 \label{eq:useMarkov}
	\\ &=\frac{T-t}{T^*-(1+\l)t} \left( \E\left[S(T^*,T^*) \left|\, S((1+\l)t,T^*)=x_1+k^{-2}x_2\right.\right]-k^{-2}x_2\right) \nonumber
	\\ &\qquad\qquad +  \frac{T^*-(T-t)}{T^*-(1+\l)t}x_1. \label{eq:useexp}
\end{align}
Equation (\ref{eq:usereorder}) holds since reordering the increments of an LRB gives an LRB with same law,
(\ref{eq:useMarkov}) follows from the Markov property,
and (\ref{eq:useexp}) follows from (\ref{eq:expectation}).
We also have
\begin{equation}
	\E\left[S(T^*,T^*) \left|\, S((1+\l)t,T^*)=x_1+k^{-2}x_2\right.\right]=\int_{0}^{\infty} z \, p_t(z) \d z,
	\label{eq:key_int}
\end{equation}
where
\begin{multline}
	p_t(z)=\1_{\{z>x_1+k^{-2}x_2\}}K^{-1} \left( \frac{z}{z-(x_1+k^{-2}x_2)} \right)^{3/2}
			\\ \times \,	\exp\left( -\half c^2 \left(\frac{(T^*-(1+\l)t)^2}{z-(x_1+k^{-2}x_2)}-\frac{T^{*2}}{z} \right) \right) p(z),
\end{multline}
and $K$ is a constant chosen to normalise the density.
Similarly, the best-estimate ultimate loss of $\{\xi_{tT}^{(2)}\}$ is
\begin{multline}
	U_{tT}^{(2)}= k^2\frac{T^*-(T-t)}{T^*-(1+\l)t}\left( \E\left[S(T^*,T^*) \left|\, S((1+\l)t,T^*)=x_1+k^{-2}x_2\right.\right]-x_1\right) 
	\\ +  \frac{T-t}{T^*-(1+\l)t}x_2.
\end{multline}
To compute $U_{tT}^{(1)}$ and $U_{tT}^{(2)}$ we need to perform at most two one-dimensional integrals: the integral we need is (\ref{eq:key_int}), but $p_t(x)$ includes a normalising constant $K$, which is found by evaluating a second integral.
We are saved the complication of performing double integrals. To extend these results to higher dimensions we can split the ``master" process $\{S_{tT}\}$ into more than two subprocesses.
Regardless of the number of subprocesses (i.e.~paid-claims processes), the best-estimate ultimate losses can be computed by performing at most two one-dimensional integrals.
This makes such a multivariate model computationally efficient.


\begin{acknowledgments}
\noindent The authors are grateful to D.~C.~Brody, H.~B\"uhlmann, A.~J.~G.~Cairns, M.~H.~A.~Davis, R. Norberg, G.~Peskir, and participants at the AMaMeF Conference, \AA lesund, Norway (May 2009), the Research in Options Conference, B\'uzios, Rio de Janeiro (November 2009), and the Sixth World Congress of the Bachelier Finance Society, Toronto (June 2010), where drafts of this paper were presented, for helpful discussions.
This work was carried out, in part, while AM was based at the Department of Mathematics, King's College London, at the Department of
Mathematics, ETH Z\"urich, and at the Institute for Economic Research, Kyoto University, and while EH and LPH were based at Imperial College London. 
EH~acknowledges the support of an EPSRC Doctoral Training Grant.
LPH~acknowledges support from Lloyds TSB, Shell International, the Aspen Center for Physics, and the Fields Institute, Toronto.
AM~acknowledges support from the African Collaboration for Quantitative Finance and Risk Research (ACQuFRR), University of Cape Town. 
\end{acknowledgments}


\vskip 15pt \noindent {\bf References}.


\end{document}